\documentclass[12pt, leqno]{article}
\usepackage{graphicx}
\usepackage{amsfonts}
\usepackage{amsmath}
\usepackage{rotating}
\usepackage{enumerate}
\usepackage{color}
\usepackage{subfig}

\usepackage{authblk}

\renewcommand{\baselinestretch}{1.1}
\usepackage{amssymb, latexsym,amsthm}
\usepackage{multirow,array,epsfig}
\usepackage{verbatim}

\newtheorem{theorem}{Theorem}

\newtheorem{assumption}{Assumption}
\newtheorem{lemma}{Lemma}

\newcommand\bSig{\mbox{\boldmath${\Sigma}$}}

\newcommand\bC{{\bf C}}

\newcommand\bI{{\bf I}}

\newcommand\bL{{\bf L}}

\newcommand\mbR{{\mathbb R}}

\newcommand\bV{{\bf V}}

\newcommand\bX{{\bf X}}

\newcommand\bZ{{\bf Z}}

\newcommand\cD{{\mathcal D}}

\newcommand\cH{{\mathcal H}}
\newcommand\cI{{\mathcal I}}

\newcommand\cS{{\mathcal S}}

\newcommand\hv{{\hat v}}

\newcommand\ip{{i^\prime}}
\newcommand\jp{j^\prime}

\newcommand\np{n^\prime}

\newcommand\eg{\varepsilon}

\DeclareMathOperator{\E}{E}

\newcommand\lp{\left (}
\newcommand\rp{\right )}
\newcommand\lb{\left [}
\newcommand\rb{\right ]}
\newcommand\lbr{\left \{}
\newcommand\rbr{\right \}}
\newcommand\lip{\left \langle}
\newcommand\rip{\right \rangle}

\usepackage{soul} 
\usepackage[sort, numbers]{natbib}

\begin{document}
\title{A randomness test for  functional panels}

\author[1]{Piotr Kokoszka \thanks{Piotr.Kokoszka@colostate.edu}}
\author[2]{Matthew Reimherr \thanks{\textit{Corresponding author}, mreimherr@psu.edu}}
\author[3]{Nikolas W\"olfing \thanks{Woelfing@zew.de, funding provided by the Helmholtz Association through the Helmholtz Alliance Energy-Trans is gratefully acknowledged.}}
\affil[1]{Colorado State University}
\affil[2]{Pennsylvania State University} 
\affil[3]{Centre for European Economic Research, ZEW Mannheim}

\date{}

\maketitle

\begin{abstract}
  Functional panels are collections of functional time series, and
  arise often in the study of high frequency multivariate data.  We
  develop a portmanteau style test to determine if the cross--sections
  of such a panel are independent and identically distributed. Our
  framework allows the number of functional projections and/or the
  number of time series to grow with the sample size.  A large sample
  justification is based on a new central limit theorem for random
  vectors of increasing dimension.  With a proper normalization, the
  limit is standard normal, potentially making this result easily
  applicable in other FDA context in which projections on a subspace
  of increasing dimension are used.  The test is shown to have correct
  size and excellent power using simulated panels whose random
  structure mimics the realistic dependence encountered in real panel
  data.  It is expected to find application in climatology, finance,
  ecology, economics, and geophysics.  We apply it to Southern Pacific
  sea surface temperature data, precipitation patterns in the
  South--West United States, and temperature curves in Germany.
\end{abstract}



\section{Introduction} \label{s:int}
We define a {\em functional panel} as a
stochastic process  of the form
\begin{equation} \label{e:panel}
\bX_n(t) = [ X_{1, n}(t), \ldots, X_{I, n}(t)]^\top, \ \ \ \
1 \le n \le N,
\end{equation}
where each $X_{i, n}$ is a function of time $t$.  The dimension $I$ can increase
with the series length $N$, with  examples discussed below.
For the applications that motivate the present research, it is enough
to think of the $X_{i, n}$ as  curves defined on the same time
interval, but in principle, functions on more general domains, e.g.,
volumes  or surfaces, can be considered.  The discrete time index $n$
refers to a unit like a day, week or year.  The index $t$ is the
continuous time argument of the function $X_{i, n}$. The index $i$
refers to the $i^\text{th}$ time series in the panel. This paper develops a
test of the null hypothesis

\medskip

$H_0$: the random elements $\bX_n, \ 1 \le n \le N,$ are independent
and identically distributed.

\medskip

\noindent Our test is designed to detect serial dependence, and we assume stationarity across $n$ even under the alternative. 

To illustrate the functional panel concept, Figure~\ref{f:ninoregions}
shows four curves, $I=4$, for a fixed $n$.  The index $n$ refers to
years, and the four curves describe the sea surface temperature in
four regions used to measure the El Ni\~{n}o climatic phenomenon.
Figure~\ref{f:SantaCruz} shows another example, now with $i$ fixed.
The data point $X_{i,n}(t)$ is the log--precipitation at location $i$
on day $t$ of year $n$. The construction of this series is explained
in detail in Section~\ref{s:fsp}.  Data structures of this type are
very common in climate studies; $X_{i, n}(t)$ can be total
precipitation or maximum temperature on day $t, \ 1 \le t \le 365,$ of
year $n$ at location $i$ in some region.  In such climate
applications, $I$ is comparable to $N$ because records often start at
the end of the 19$^{\text{th}}$ or towards the middle of the
20$^{\text{th}}$century, thus, they are about 60 to 120 years long
($N\approx 60\text{ to }120$), and there are several dozen measurement
stations in a region ($I\approx 40 \text{ to }120$).  (The United
States Historical Climatology Network -- USHCN -- contains weather
data collected at 1,218 stations across the 48 contiguous states,
starting from ca. 1900.)

\begin{figure}[hbt]
\begin{center}
\includegraphics[width=0.8\textwidth]{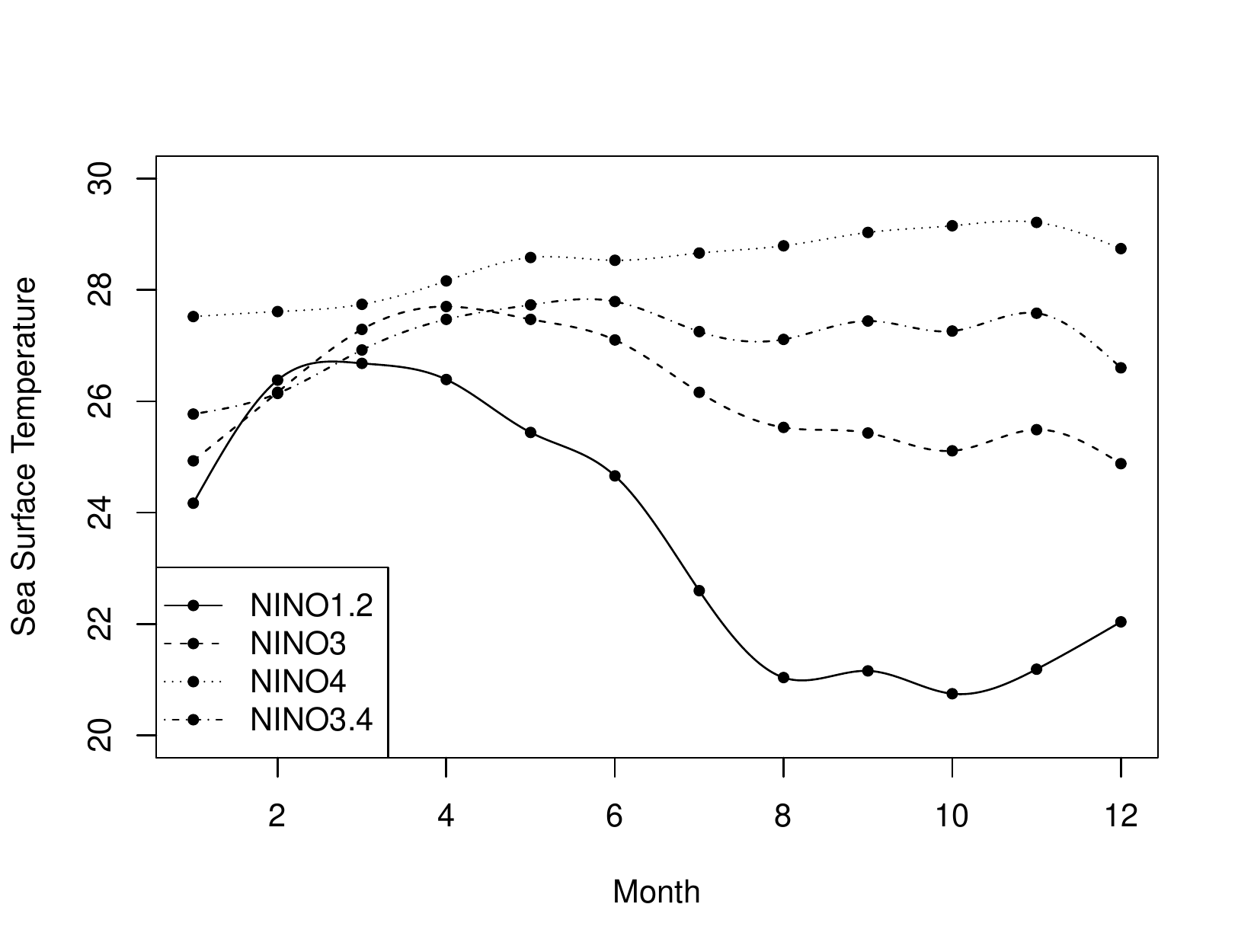}
\end{center}
\caption{Sea Surface Temperature curves of El Ni\~{n}o regions in 2012}
\label{f:ninoregions}
\end{figure}

\begin{figure}[ht]
\begin{center}
\includegraphics[width=\textwidth]{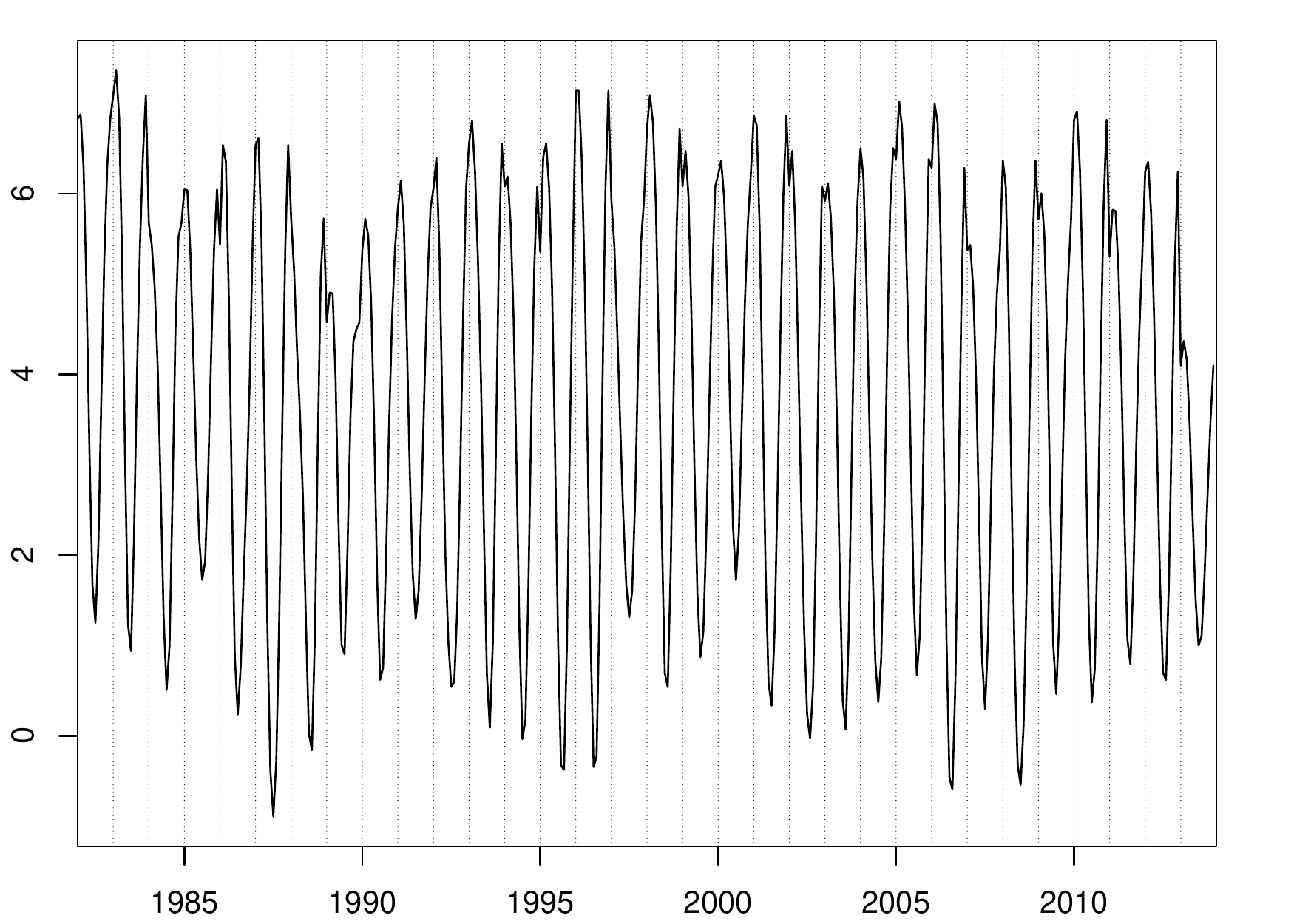}%
\end{center}
\caption{Smoothed log--precipitation, Santa Cruz, California, 1982 to 2013.}
\label{f:SantaCruz}
\end{figure}

Climate data do not exhaust possible applications.  Intraday
financial data typically come in panels. For example, $X_{i,n}(t)$
can be the exchange rate (against the US dollar) of currency $i,\ 1
\le i \le I,$ at minute $t$ of the $n^\text{th}$ trading day.  Panels of
exchange rates contain information on the intraday strength of the US
dollar. Corporate  bond yield curves are
large panels because a bond portfolio includes hundreds of companies,
$I\sim 10^3$; in economic studies, government bond yields curves form
small panels because only a few countries are considered to assess
risk in a region, see, e.g., \citet{hardle:majer:2015}.
At the intersection of climate and financial panels,
\citet{hardle:osipienko:2012}  use  a functional panel framework
in which $i$ refers to a spatial location, and the interest lies in
pricing a financial derivative product whose value depends on the
weather at location $i$. Modeling electricity data involves functional
panels indexed by regions or power companies with the daily index $n$,
see \citet{liebl:2013} for an overview. Daily pollution (particulate,
oxide or ozone) curves at several locations within a city form a functional
panel of moderate size.

In these examples, the dependence between the $X_{i,n}, 1 \le i \le
I$, for fixed $n$, is strong, and, generally, the temporal dependence,
indexed by $n$, cannot be neglected. In specific applications, this
dependence is modeled by deterministic trends or periodic functions
(climate data) or by common factors (financial data).  To validate a
model, it is usual to verify that  residual curves computed in some
manner form a random sample.
(See, e.g., \citet{kowal:matteson:ruppert:2014} for a model applied 
to functional panels of government yield curves and neurological 
measurements with an explicit residual iid assumption.)
It is thus important  to develop a test of
randomness, i.e., to test the null hypothesis $H_0$ stated above.  Such a
test could be viewed as analogous to tests of randomness which are
crucial in time series analysis, see, e.g.,  Section 1.6 of
\citet{brockwell:davis:2002}. They can be applied to original or
transformed data, or to model residuals.  The purpose of this paper is
to develop a suitable test for functional panels.
Before discussing our approach, we provide
some historical background.  Our methodology builds on the well-established 
paradigm of testing for randomness in time series which
can be traced back to the work of \citet{box:pierce:1970}, which
was followed by a number of influential contributions including
\citet{chitturi:1976,hosking:1980,ljung:box:1978} and \citet{mcleod:1978}.  
These tests use as a starting point
 the asymptotic distribution of the sample autocorrelations
of a white noise: the  $\hat\rho_h$ are approximately independent
normal random variables with mean zero and variance $1/N$, where 
$N$ is the sample size. 
Therefore, $N\sum_{h=1}^H {\hat\rho_h}^2$ is approximately
chi--square with $H$ degrees of freedom.
This research is
now reported in textbook expositions including
\citet{brockwell:davis:1991}, \citet{li:2004}, and
\citet{lutkepohl:2005}.  More recent contributions include
\citet{fisher:gallagher:2012} and \citet{pena:rodriguez:2002}.
For a single functional time series, a randomness test was derived by
\citet{gabrys:kokoszka:2007} and elaborated on by
\citet{horvath:huskova:rice:2013} and \citet{jiofack:nkiet:2010}.  In the context of {\em scalar} panel
data, the only work we are aware of is \citet{fu:etal:2002} who
define residual autocorrelations in the autoregressive panel model of
\citet{hjellvik:tjostheim:1999} as
\[
\hat r_h = \frac{\sum_{i=1}^I \sum_{n=1}^{N-h} \hat\eg_{i, n+h}
\hat\eg_{i,n}}
{\sum_{i=1}^I \sum_{n=1}^{N-h}\hat\eg_{i,n}^2},
\]
where the $\hat\eg_{i,n}$ are appropriately defined residuals.  In
their asymptotic setting, the number of temporal points, $N$, is
fixed, and the number of time series, $I$, increases to infinity. They
show that for any fixed $H$, the vector $[\hat r_1, \ldots,
\hat r_H]^\top$ is asymptotically normal with the asymptotic covariance
matrix that can be estimated. By constructing a suitable  quadratic
form, they derive a portmanteau test statistic whose asymptotic
distribution is $\chi^2_H$. There is at present no randomness test
for functional panels, and it is our objective to derive a practically
useable test which is supported by asymptotic arguments.
\citet{hsiao:2003} provides an excellent account of the methodology
for scalar panel data.

We reduce the dimension of the functions $X_{i,n}$ by using
projections on functional principal components. Denote the number of
such projections for the $i^\text{th}$ series in the panel by $p(i)$. The
total number of scalar time series we must consider is thus
$p=\sum_{i=1}^I p(i)$. For climate applications discussed above, $p$
can approach several hundred.  It is thus natural to consider
asymptotics with $p$ increasing to infinity with $N$.  Our theory
applies to cases of fixed $p(i)$ and increasing $I$, fixed $I$ and
increasing $p(i)$, or both increasing with $N$.  In this framework, we
show that it is possible to construct a test statistic that is
asymptotically standard normal.  The work of \citet{fu:etal:2002}
can be viewed as considering $p(i)=1$, $N$ fixed, and $I\to\infty$.
Our setting is thus quite different, and requires a new asymptotic
framework. Despite some theoretical complexity, our approach leads to
a test whose asymptotic null distribution is standard normal, and
which can be algorithmically implemented. It is therefore hoped that
our work will find application in the analysis of functional panels of
the type specified above. In particular, it could motivate the
development of suitable change point tests that target  more
specific alternatives, \citet{aston:kirch:2012AAS} and \citet{zhang:shao:hayhoe:wuebles:2011} consider the case of $I=1$.

Since the goal of our methodology is to check the independence
assumption, our tests are based on the usual functional principal
component scores; functional principal components form the optimal
basis under the null hypothesis. They  have an established place in FDA
research with readily available {\tt R} and {\tt matlab}
implementations.  In principle, other basis systems could be used,
especially those custom--developed for time series of functions, see, e.g.,
\citet{hormann:2015} and \citet{panaretos:2013}, or even those
going beyond linear dimension reduction, see \citet{li:song:2016}.
In each of these cases, our general approach could be applied to the
resulting scores, but new asymptotic justifications and numerical 
implementations  would have to be developed.

The remainder of the paper is organized as follows.  In
Section~\ref{ss:results}, we formalize the asymptotic framework,
derive the test statistic and establish its asymptotic normality.
Section~\ref{ss:fsi} describes the practical implementation of the
test procedure in algorithmic steps. Section \ref{s:fsp} illustrates
the application of the test on three climate data sets which form
functional panels: sea surface temperatures in the pacific ocean, US
regional precipitation data and temperature curves in Germany.
Section \ref{ss:sim} further examines the finite sample performance of
our test by applying it to simulated data which resemble the above
mentioned data sets.  The proofs of the asymptotic results are
presented in Section \ref{s:multi} and in the supplemental material.
In addition to these mathematical calculations, the supplemental
material contains a zipped folder containing the complete {\tt R}
code, a corresponding {\tt README} file and the data sets.

\section{Testing procedure} \label{s:test}

\subsection{Assumptions and large sample
results} \label{ss:results}

We assume that all functions have been rescaled so that their
domain is the unit interval $[0,1]$. We also assume that
they have mean zero: $\E X_{i,n}(t) = 0$ for almost all $t\in[0,1]$.
In practice, the mean is removed by subtracting the sample mean, see Section~\ref{ss:fsi}, so that the functional time series forming the
panel each have sample mean zero. Subtracting the sample
mean introduces  additional terms of the
order $O_P(N^{-1})$, and so does not affect the limiting distribution.

Denote by $L^2= L^2([0,1])$ the Hilbert space of square integrable
functions with the usual inner product $\langle \cdot, \cdot \rangle$
and the norm $\| \cdot \|$ it generates.  The assumptions and the
definition of the test statistic involves the Kronecker product, and
its properties are heavily used in the proofs. Readers are referred to
\citet{graham:1981} for a very useful exposition.  In the context
of matrices and vectors, we take $\otimes$ to be the usual Kronecker
product.  Between two functions or operators $x$ and $y$, we take $x
\otimes y$ to be the operator $\langle x, \cdot \rangle y$.  We will
often not distinguish notationally between the cases, as it will
always be clear from the context which we mean.  Further details will
be provided as needed.  By $| \cdot |$ we denote the Euclidean norm of
a vector.

Our first assumption states the functions forming the panel
are in $L^2$ and have uniformly bounded fourth moments.

\begin{assumption} \label{a:1}
Assume that $\{\bX_n\}$ is a {\em zero mean}
sequence of random
functional vectors taking values in $\{L^2\}^I$.  Furthermore, assume
that there exists a constant $M$ such that
\[
\E \| X_{i,n}\|^4 \leq M <  \infty,
\quad i=1, \ldots, I,  \quad n=1,\ldots,N.
\]
\end{assumption}

Our second assumption connects the rate of growth of $I$ and the
$p(i)$ to the rate of decay of the gaps between the eigenvalues
of individual series and the rate of decay of the eigenvalues of
the whole panel. Assumptions of this type go back at least to the work
of \citet{dauxois:1982}. 
To the best of our knowledge, only the case of a single functional
series or sample, possibly with explanatory variables or functions,
has been considered, see \citet{cai:hall:2006}, \citet{crambes:kneip:sarda:2009},
\citet{fremdt:2014}, \citet{hall:muller:wang:2006}, and \citet{paul:peng:2009}, 
among many others.
The complexity of our Assumption~\ref{a:pn2} is due to the
panel structure of the data. To formulate it,  define
\[
{\bf X}_{in} = [ X_{1in},  \ldots, X_{p(i)in}]^\top, \ \ \ \
X_{jin} = \langle X_{i,n}, v_{i,j} \rangle
\]
and column vectors
$
{\bf X}_n^\star
= [{\bf X}_{1n}^\top, \ldots, {\bf X}_{In}^\top]^\top
$
of length $p_N:=p= \sum_{i=1}^{I} p(i)$.
The panel $\lbr {\bf X}_n^\star\rbr$ is thus an approximation
of dimension $p_N$ to the functional panel  $\lbr {\bf X}_n \rbr$
given by \eqref{e:panel}.
 Let
\[
\bC_{0,N} = \E \lb\bX_n^\star \bX_n^{\star \top} \rb
\]
be the $p_N\times p_N$ covariance matrix whose eigenvalues are
$\gamma_1 \ge  \ldots \ge \gamma_{p_N}$.  Denote by
$\lambda_{i,1} > \lambda_{i,2}>  \ldots $ the eigenvalues of the
covariance operator $\E[ X_{i,1} \otimes X_{i,1}]$ and define
\[
\Gamma_N = \sum_{i,\ip=1}^I \sum_{j = 1}^{p(i)}
\sum_{\jp = 1}^{p(\ip)}
(\alpha_{i,j}^{-1} + \alpha_{\ip, \jp}^{-1})^2,
\]
where $\alpha_{i,1} = \lambda_{i,1} - \lambda_{i,2}$ and for $j \geq 2$,
$\alpha_{i,j}
= \min\{\lambda_{i, j-1} - \lambda_{i,j},
\lambda_{i,j} - \lambda_{i, j+1}  \}$.

\begin{assumption} \label{a:pn2} Assume that the sequence $p_N$
 is such that $p_N \to \infty$ and
\[
N^{-1/2}   p_N^{-1}   \gamma_{p_N}^{-3}  I^3\Gamma_N^{1/2} \to 0.
\]
(The number of panels, $I$, can either stay fixed or tend to infinity.)
\end{assumption}

Assumption~\ref{a:pn2} has the following interpretation.  The first
two terms, $N^{-1/2} p_N^{-1}$, indicate the rate at which information
accumulates as $N \to \infty$, while the third and fourth terms,
$\gamma_{p_N}^{-3} I^3$, indicate the rate at which the panel
structure detracts information (with $\gamma_{p_N}$ governing the
correlation between series).  The last term $\Gamma_N^{1/2}$
incorporates the spacing of the eigenvalues and is common in
asymptotics with an increasing number of projections.  A more readily
interpretable form of this assumption is stated in
Section~\ref{ss:I=1} for the case of a single time series ($I=1$).  We
do not impose any specific dependence structure, and prefer to use a
general, admittedly rather technical, Assumption~\ref{a:pn2}.  An
alternative approach would be to impose some temporal dependence
structure, e.g.,  as in \citet{jirak:2015}, and establich analogous
results under such assumptions.  Instead, we give a brief example to
help shed further light on Assumption \ref{a:pn2}.

Assume that each element of the panel has the same covariance operator
so that $\lambda_{i,j} \equiv \lambda_j$ and $\alpha_{i,j} \equiv
\alpha_j$ for all $i$ and $j$.  In this case, it makes sense to also
assume that $p_i \equiv p$ so that $p_N = I p$.  Collect the
$\lambda_j$ into a diagonal matrix $\Lambda$.  Furthermore, assume
that the panels are independent so $\bC_{0,N} = \Lambda \otimes \bI_{I
  \times I}$, where $\otimes$ denotes the Kronecker product and
$\bI_{I\times I}$ the $I \times I$ identity.  This then implies that
$\gamma_{p_N} = \lambda_{p}$.
 
We now assume explicitly that $\lambda_j = j^{-\alpha}$ and $\lambda_j - \lambda_{j+1} = j^{-\alpha - 1}$.  This implies that
\begin{align*}
\Gamma_N & \leq I^2 \sum_{j} \sum_{j^\prime} \left[j^{\alpha + 1} + (j')^{\alpha + 1}  \right]^2 \\
& = I^2 \left[ 2 p \sum_{j=1}^p j^{2\alpha + 2} + 2   \left(\sum_{j=1}^p j^{\alpha + 1}\right)^2 \right] \\
& \approx  I^2 \left[ 2  p^{2\alpha +4} (2 \alpha + 3)^{-1} + 2 p^{2 \alpha + 4} ( \alpha + 2)^{-2} 
\right] \sim I^{2} p^{2 \alpha + 4}.
\end{align*}
Here $\approx$ means the limit of their ratio tends to 1, while $\sim $ means the limit of their ratio is a finite nonzero constant.  So then we have
\[
N^{-1/2}   p_N^{-1}   \gamma_{p_N}^{-3}  I^3\Gamma_N^{1/2} 
\sim N^{-1/2} p^{-1} I^{-1} p^{3 \alpha} I^{3} I p^{\alpha + 2}
= N^{-1/2} p^{4 \alpha +1} I^2.
\]
The parameter $\alpha$ is usually viewed as the smoothness of the $X_{i,n}$ processes.  We can see that for rougher processes, we can actually take larger panels and more principal components, since $\alpha$ will be smaller in these cases.  For example, $\alpha = 2$ for Brownian motion.  The same calculations will show that in the single panel case of Section \ref{ss:I=1}, the rate becomes
\begin{align*}
\frac{N^{-1/2} \sum_{j=1}^p \alpha_j^{-1} }{\lambda_p^2}
\sim N^{-1/2} p^{3 \alpha +2}.
\end{align*}
Since it must be the case that $\alpha > 1$, we can see the price we pay for the lack of structure in the panel as $ 3\alpha + 2 < 4\alpha + 1$.  This price increases for smoother processes, i.e., larger $\alpha$.

We now proceed to define the test statistic.
Let  $\hat v_{i, j}$  be the $j^\text{th}$ estimated functional principal
component (EFPC) of the $i^\text{th}$ functional time series, see, e.g.,
Chapter 3 of \citet{HKbook}. Set
\[
\widehat {\bf X}_{in} = [\widehat X_{1in},  \ldots, \widehat X_{p(i)in}]^\top, \ \ \ \
\widehat X_{jin} = \langle X_{i,n}, \hat v_{i,j} \rangle.
\]
Next, we form column vectors of length $p_N:=p= \sum_{i=1}^I p(i)$
given by
\[
\widehat {\bf X}_n
= [\widehat {\bf X}_{1n}^\top, \ldots, \widehat {\bf X}_{In}^\top]^\top.
\]
To form a portmanteau test statistic using the $\widehat \bX_n$,
we introduce
\begin{align*}
\widehat \bV_h = N^{-1} \sum_{n=1}^{N-h}
\widehat \bX_n \otimes \widehat \bX_{n+h}; \ \ \ \
\widehat \bC_0 =  N^{-1} \sum_{n=1}^N
\widehat \bX_n \widehat \bX_{n}^\top.
\end{align*}
Observe that $\widehat \bV_h$ is a column vector of length $p_N^2$ and
$\widehat \bC_0 \otimes \widehat \bC_0$ is a $p_N^2\times p_N^2$
symmetric matrix. The test statistic is defined by
\begin{align*}
\widehat Q_N & = N \sum_{h=1}^H \widehat \bV_h^\top
(\widehat \bC_0 \otimes \widehat \bC_0)^{-1} \widehat \bV_h.
\end{align*}
The summation limit $H$ plays the same role as the maximal number
of lags in the usual Box--Pierce--Ljung type statistics.
It is fixed in the asymptotic theory.

Our first result states that $\widehat Q_N$ is asymptotically  normal
under $H_0$, i.e.,\ when the data are iid.

\begin{theorem} \label{t:multi}
If Assumptions \ref{a:1} and \ref{a:pn2} hold, then under $H_0$,
\[
\frac{\widehat Q_N - p_N^2 H}{p_N \sqrt{2H}} \overset{\cD}{\to} \mathcal{N}(0,1).
\]
\end{theorem}
Theorem~\ref{t:multi} is proven in  Appendix~\ref{s:multi}. The proof 
involves a sequence of vectors of
projections of increasing dimension.  In
\citet{cardot:fms:2003} and \citet{horvath:huskova:rice:2013},
this problem is avoided by making extensive use of the Prokhorov--Levy
metric.  However, such a technique is limited due to the difficulty of
incorporating dimension into any Berry--Esseen type convergence result,
which typically rely on highly complex smoothing arguments.
Furthermore, such an approach typically does not yield results as
sharp as proving the CLT directly due to the way they depend on
dimension.  In contrast, our approach adds no additional assumptions
beyond those needed to replace the estimated eigenvalues and
eigenfunctions with their theoretical counterparts.  We therefore view
the following theorem, which establishes the asymptotic normality of
general quadratic forms based on autocorrelations, as an important
contribution of this paper.

\begin{theorem}\label{t:norm}
Let $\{\bZ_{n,N}, 1 \le n \le N\}$ be an array of random vectors with
$\bZ_{n,N} \in \mbR^{p_N}$.  For each $N$,
 assume that $\bZ_{1,N}, \dots, \bZ_{N,N}$ are iid and that
\begin{equation} \label{e:m-cond-Z}
\E[\bZ_{1,N}] ={ \bf{0}} \quad \mbox{and}
\quad \E[\bZ_{1,N} \bZ_{1,N}^\top]
= \bI_{p_N}.
\end{equation}
If, as $N \to \infty$,
\begin{equation} \label{e:p-cond-Z}
p_N \to \infty, \ \ \ p_N N^{-2/3} \to 0, \ \ \
{\rm and} \ \ \   N^{-1/2} \E| \bZ_{1,N}|^4 \to 0,
\end{equation}
then for $H$ fixed
\[
 \frac{N^{-1}\sum_{h=1}^H \left| \sum_{n=h+1}^{N} \bZ_{n-h,N}
\otimes \bZ_{n,N} \right|^2 - p_N^2 H}{p_N \sqrt{2H}}
\overset{\cD}{\to} \mathcal{N}(0,1).
\]
\end{theorem}
The proof of Theorem~\ref{t:norm} is presented in the supplemental
material.
Limit results for random vectors with the dimension increasing with
the sample size appear in the asymptotic theory for empirical
likelihood, see, e.g., \citet{hjort:mckeague:vankeilegom:2009}
and \citet{peng:schick:2013}. The Central Limit Theorem
established by \citet{peng:schick:2012} is motivated by
such theory. Using the notation of Theorem~\ref{t:norm},
a corollary to their main result can be stated as
\[
\frac{N^{-1}\left | \sum_{n=1}^N \bZ_{n, N} \right |^2 - p_N}{\sqrt{2p_N}}
\overset{\cD}{\to}  \mathcal{N}(0,1).
\]
Their focus is not on the lagged Kronecker products, but on the case
where $\E[\bZ_{1,N} \bZ_{1,N}^\top]$ is a general covariance matrix (i.e.,\ not the identity),
and the centering is with respect to its trace.  Other results on CLT
convergence rates which incorporate dimension are given in
\citet{senatov:1998}.

We conclude this section with a general framework under which the test
rejects the null.  If the sequence is stationary and 
weakly dependent and at least
one element of the panel exhibits nonzero correlation with another
element (at some lagged time index), then the test will reject with
power approaching one. As a specific assumption for stationarity 
and weak dependence we use the concept of $L^p$--$m$--approximability, 
see \citet{HKbook}, Chapter 16. 

\begin{assumption}{(Alternative Hypothesis)} \label{a:HA}
 Assume that $\{(X_{i,1},\dots,X_{i,N})\}$ is a stationary $L^4$-m
 approximable sequence (for each $i$) and that there exists a nonempty
 subset of indices $\cH^\star \subset \{1,\dots,H\}$, and for each $h
 \in \cH^\star$ a nonempty subset of pairs of indies $\cI_h^\star
 \subset \{1, \dots, I\}^2$ such that 
 and $j \in \{1,\dots, I\}$ such that 
\[
	\left(\E[\langle X_{n}, v_{i}
		\rangle \langle X_{n+h}, v_{j} \rangle ]\right)^2
		\geq R > 0,
\] 
for all $h \in \cH^\star$ and $(i,j) \in \cI_h^\star$.
\end{assumption}

\begin{theorem} \label{t:multi:HA}
If Assumptions \ref{a:1}, \ref{a:pn2}, and \ref{a:HA} hold, then 
\[
\frac{\widehat Q_N - p_N^2 H}{p_N \sqrt{2H}} 
\gtrsim \gamma_1^{-2} R N (1+ o_P(1)) \sum_{h \in \cH^\star} | \cI_h^\star|
\overset{P}{\to} \infty.
\]
\end{theorem}
\noindent Theorem~\ref{t:multi:HA} is proven in  Appendix~\ref{s:multi}

In the next section, we consider the case of a single series to illustrate 
 our assumptions. We emphasize that $H$ is assumed 
to be fixed, but can be arbitrarily large. Asymptotic under $H$ diverging 
to infinity (for a single series) were investigated by
\citet{horvath:huskova:rice:2013}. It should, in principle,  be possible 
to let the number of panel series, $I$, the number of projections
$p$ and the maximal lag $H$ tend simultaneously to infinity, but we do not 
develop such a more complex theory here. We thus stay within the 
framework  of traditional time series analysis where asymptotics are
derived for a finite number of lags, see, e.g, Chapter 7 of 
\citet{brockwell:davis:1991}. 

\subsection{Case of $I=1$ (a single functional series)} \label{ss:I=1}
To provide a more tangible intuition behind the form of the statistic
$\widehat Q_N$, we discuss the simpler scenario where we only have one
time series, i.e., $I=1$.  Define
\[
\Delta_{N,h} = N^{-1/2} \sum_{n=1}^{N-h}  X_n \otimes X_{n+h}.
\]
The autocovariance operator $\Delta_{N,h}$ is Hilbert--Schmidt. Recall
 that Hilbert--Schmidt operators form a separable Hilbert space with
the inner product
\begin{equation} \label{e:HS-prod}
\lip \Psi_1, \Psi_2 \rip_\cS = \sum_{k} \lip \Psi_1(e_k), \Psi_2(e_k) \rip,
\end{equation}
where $\lbr e_k \rbr$ is any orthonormal basis, see, e.g., Chapter 2 of
\citet{HKbook}.  A direct application of this definition with the
functions $\hv_k$ (extended to a complete system) shows that
\[
\langle \Delta_{N,h},  \hat v_j \otimes \hat v_{\jp} \rangle_\cS
= N^{-1/2} \sum_{n=1}^{N-h} \widehat X_{jn} \widehat X_{\jp, n+h}.
\]
Therefore, by Lemma 7.1 of \citet{HKbook}, the statistic
$\widehat Q_N$ can be expressed as
\[
\widehat Q_N = \sum_{h=1}^H \sum_{j, \jp=1}^{p_N}
\frac{\langle \Delta_{N,h},  \hat v_j \otimes \hat v_{\jp} \rangle_\cS^2}
{\hat \lambda_j \hat \lambda_{\jp}}.
\]
The summands are the squares of the sample cross--correlations of
all projections
under consideration. These are added over all projections and all lags up
to lag $H$.

In the case of a single time series, Assumption~\ref{a:pn2} can be replaced
by a more interpretable assumption:
\begin{assumption}\label{a:pn}
We assume that the sequence $p_N$ is nondecreasing,
$p_N \to \infty$, and satisfies
\[
\frac{N^{-1/2} \sum_{j=1}^{p_N} \alpha_j^{-1}  }
{ \lambda_{p_N}^2} \to 0,
\]
where $\alpha_1 = \lambda_1 - \lambda_2$ and for $j \geq 2$,
$\alpha_j
= \min\{\lambda_{j-1} - \lambda_j, \lambda_j - \lambda_{j+1}  \}$.
\end{assumption}
Assumption \ref{a:pn} quantifies the intuition that $p$ should
increase to infinity at a rate slower than $N$, depending on the rate
of decay of the eigenvalues $\lambda_j$ and the gaps
between them. Direct verification shows
that if the $\lambda_j$ decay exponentially fast, then $p$ must increase
 at a rate slower than $\ln N$. If the $\lambda_j$ decay like a
power function, Assumption~\ref{a:pn} will hold if $\ln(p)/\ln(N) \to
0$.
The proof of Theorem~\ref{t:multi} for $I=1$ is presented in the
supplemental material. It is less abstract than the general proof
in Appendix~\ref{s:multi};  its study may facilitate the understanding
of the general case.

\subsection{Details of implementation} \label{ss:fsi}
In this section, we provide a step by step description of the testing
procedure. All steps listed below can be easily implemented in {\tt R}
or {\tt Matlab} using basic routines and the functional principal
component tool box (see \citet{ramsay:hooker:graves:2009}).
Supplemental material contains a ready to use {\tt R}
function implementing the test. For step 6, we provide two alternative
but equivalent procedures where the first might be more intuitive
while the second is computationally more efficient. The finite sample
bias correction in step 7 follows from an extension of the arguments
of \citet{ljung:box:1978}.  It is clearly asymptotically negligible.

\begin{enumerate}
\item Center each functional time series, i.e., compute
\[
X_{i, n}^c (t)  = X_{i,n}(t)  - \hat\mu_i(t), \ \ \
 \hat\mu_i(t) = N^{-1} \sum_{n=1}^N X_{i,n}(t).
\]
\item Calculate the eigenfunctions $\hat v_{i,j}$ and the eigenvalues
$\hat \lambda_{i,j}$ of the empirical covariance operator defined by
\[
\widehat C_i (x) = \frac{1}{N} \sum_{n=1}^N
\langle X_{i,n}^c, x \rangle X_{i,n}^c.
\]
(This step is implemented as \verb|pca.fd| in {\tt R} and as \verb|pca_fd|
in {\tt Matlab}. Both functions, by default,
 center their arguments as in step 1.)
\item For each $1 \le i \le I$,  determine $p(i)$ as the smallest $k$ for
which
\[
\frac{\sum_{j=1}^k\hat\lambda_{i,j}}
{\sum_{j=1}^{N}\hat\lambda_{i,j}} > 0.85.
\]
\item Construct  the vectors of scores
\[
\widehat {\bf X}_{in} = [\widehat X_{1in}, \ldots, \widehat X_{p(i)in}]^\top, \ \ \ \
\widehat X_{jin} = \langle X_{i,n}, \hat v_{i,j} \rangle.
\]
and the vectors
\[
\widehat {\bf X}_n
= [\widehat {\bf X}_{1n}^\top, \ldots, \widehat {\bf X}_{In}^\top]^\top.
\]
Using the vectors $\{\widehat {\bf X}_n \}$, generate the $N \times p_N$ matrix
\[
\widehat{\bf X} = [\widehat{\bf X}_1,..., \widehat{\bf X}_N]^\top
\]
and calculate the empirical covariance matrix
\[
\widehat{\bf C}_0 = N^{-1} \widehat{\bf X}^\top \widehat{\bf X} .
\]
\item
Calculate  the spectral decomposition
\[
\widehat \bC_0 = {\bf U} {\bf D} {\bf U}^\top,
\]
where ${\bf D}$ is the diagonal matrix of eigenvalues and
${\bf U}$ the matrix of eigenvectors.
Let $d_1, \dots, d_{p_N}$ be the eigenvalues of $\widehat \bC_0$.
Choose a cutoff point $q$ defined as  the smallest integer for which
\[
\frac{\sum_{i=1}^q d_i}{\sum_{i=1}^{p_N} d_i} \geq 0.85.
\]
Set  ${\bf D}^{-1}[i,i] = d_i^{-1}$ if $1\leq i \leq q$ and
zero otherwise. Calculate the generalized inverse
\[
\widehat \bC_0^{-1} = {\bf U} {\bf D}^{-1} {\bf U}^\top.
\]
(Note that $\widehat \bC_0$ might be singular, e.g., when $p_N>N$.)
\item
Using the vectors $\{\widehat {\bf X}_n \}$ from step 4, compute the terms
\[
\widehat \bV_h = N^{-1} \sum_{n=1}^{N-h}
\widehat \bX_n \otimes \widehat \bX_{n+h}
\]
for $1 \leq h \leq H$, where $H$ is chosen by the user.
(A discussion of this issue is presented at the end of Section~\ref{ss:sim}.)

The test statistic
\[
\widehat Q_N = N \sum_{h=1}^H \widehat \bV_h^\top
(\widehat \bC_0 \otimes \widehat \bC_0)^{-1} \widehat \bV_h.
\]
can be calculated using 
$(\widehat \bC_0 \otimes \widehat \bC_0)^{-1}
\approx {\widehat \bC_0}^{-1} \otimes {\widehat \bC_0}^{-1}$,
where ${\widehat \bC_0}^{-1}$ is the generalized inverse from step 5.
\vspace{\baselineskip}

An alternative procedure for step 6 which avoids the Kronecker product and
therefore requires less computational resources is the following:\\
For each  $h, 1 \leq h \leq H$, take two submatrices of the matrix $\widehat{\bf X}$
defined in step 4,
\begin{align*}
\widehat{\bf X}_{-h} = [\widehat{\bf X}_1,..., \widehat{\bf X}_{N-h}]^\top \\
\widehat{\bf X}_{+h} = [\widehat{\bf X}_{1+h},..., \widehat{\bf X}_{N}]^\top
\end{align*}
and construct  the matrices
\[
	\widehat{\bf M}_h = N^{-1}  ( \widehat{\bf X}_{-h}^\top \widehat{\bf X}_{+h} )^\top ,\,\,\, h=1,\, ...,\, H.
\]
(Note that the vectorized form $\text{vec}\!\! \left(\widehat{\bf M}_h \! \right)$ is equivalent to the
vectors $\widehat{\bf V}_h$ defined before.)

Calculate  the test statistic as
\[
\widehat Q_N = N \sum_{h=1}^H
	\text{vec} \!\! \left(\widehat{\bf M}_h\right)^\top
	\text{vec} \!\! \left(
			\left(\widehat \bC_0^{-1}\right)^\top \widehat{\bf M}_h \, \widehat \bC_0^{-1}
	\right),
\]
where ${\widehat \bC_0}^{-1}$ is the generalized inverse from step 5.

\item Reject the null hypothesis  at significance level $0< \alpha < 1$, if
\[
\frac{\widehat Q_N - q^2 H \left(1-\frac{H+1}{2N}\right)}{q\sqrt{2H \left(1-\frac{H+1}{2N}\right)}}
> \Phi^{-1}(1-\alpha),
\]
where $q$ is determined in step 5 and where $\Phi^{-1}(1-\alpha)$ is the
$(1-\alpha)$th quantile of the standard normal distribution.
\end{enumerate}

\section{Applications and finite sample performance}
\label{s:fsp}
As discussed in Section~\ref{s:int}, there are many examples of
functional panels, with various temporal and cross-sectional
dependence structures, and various shapes of the curves. This 
paper focuses on methodology and theory. It is therefore not 
possible to present a simulation study which covers the wide range  
of possibly relevant scenarios. However, rather than considering 
some ad hoc artificial data generating processes (DGP), we focus on 
three real data sets taken from climate studies and then 
simulate panels whose random structure resembles the one of these 
real data sets closely. Our goal is to evaluate the performance of 
the test in realistic settings and so to provide additional 
guidance for its application.

\subsection{Application to climate data} \label{ss:data}
We consider three climate data sets with different
values of $I$ and $N$ and different levels of noise.
Each of them consists of $N$ annual curves at $I$ locations.
Before describing these data  in more detail,
we provide the following summary.
\begin{enumerate}
\item {\bf El Ni\~{n}o SST:} $N=63, \ I=4$,\  smooth.
\item {\bf US precipitation:} $N=113, \  I = 103$,\  noisy.
\item {\bf German temperature:} $N=61,\  I=42$,\  noisy.
\end{enumerate}

El Ni\~{n}o is a phenomenon of semi--periodic variation
of sea surface temperature (SST) in the southern Pacific Ocean.  The
phenomenon is measured by an index for SST variation in several
regions, generally referred to as Ni\~{n}o-1+2, Ni\~{n}o-3,
Ni\~{n}o-4, and Ni\~{n}o-3.4, see, e.g., \citet{trenberth:1997}, and
\citet{trenberth:stepaniak:2001}.  Data with monthly SST
measurements for all four regions from 1950 onwards are available
online:\\
{\tt http://www.cpc.ncep.noaa.gov/data/indices/ersst3b.nino.mth.81-10.ascii}.
Panels with $I$ comparable to $N$ arise in regional climate studies.
As an example, we use monthly precipitation data from the United
States Historical Climatology Network for all stations in California,
Arizona, and New Mexico, a region known for reoccurring precipitation
deficits. Screening for completeness yields a panel of $I=103$
stations from 1901 to 2013, $N=113$.  Increasing all records by 0.01
inch (which is the smallest discrete unit of measurement for these
data) allows us to use a log-transformation which leads to
approximately normal data.  (Any invertible transformation preserves
the iid property stated as $H_0$.)  As our last example, we consider
daily temperature data at $I=42$ weather stations in Germany over $N=
61$ years (accessible online at: {\tt www.dwd.de}).

To remove long term trends that would lead to a rejection, all data
were detrended by fitting a least squares line to observations at each
location and each month. This operation had a visible effect only on
the German temperature data, as illustrated in Figure
\ref{fig:august}. Arguably, more sophisticated methods 
of long term behavior modeling could be used, methods, where testing the 
residuals for independence might be crucial to conclude on the model 
validity. However, our objective is not a climatological study, but merely 
an illustration of statistical methodology which can be applied to 
residual curves obtained from more complex models.

\begin{figure}[hbt]%
\centering
\includegraphics[width=0.6\columnwidth]{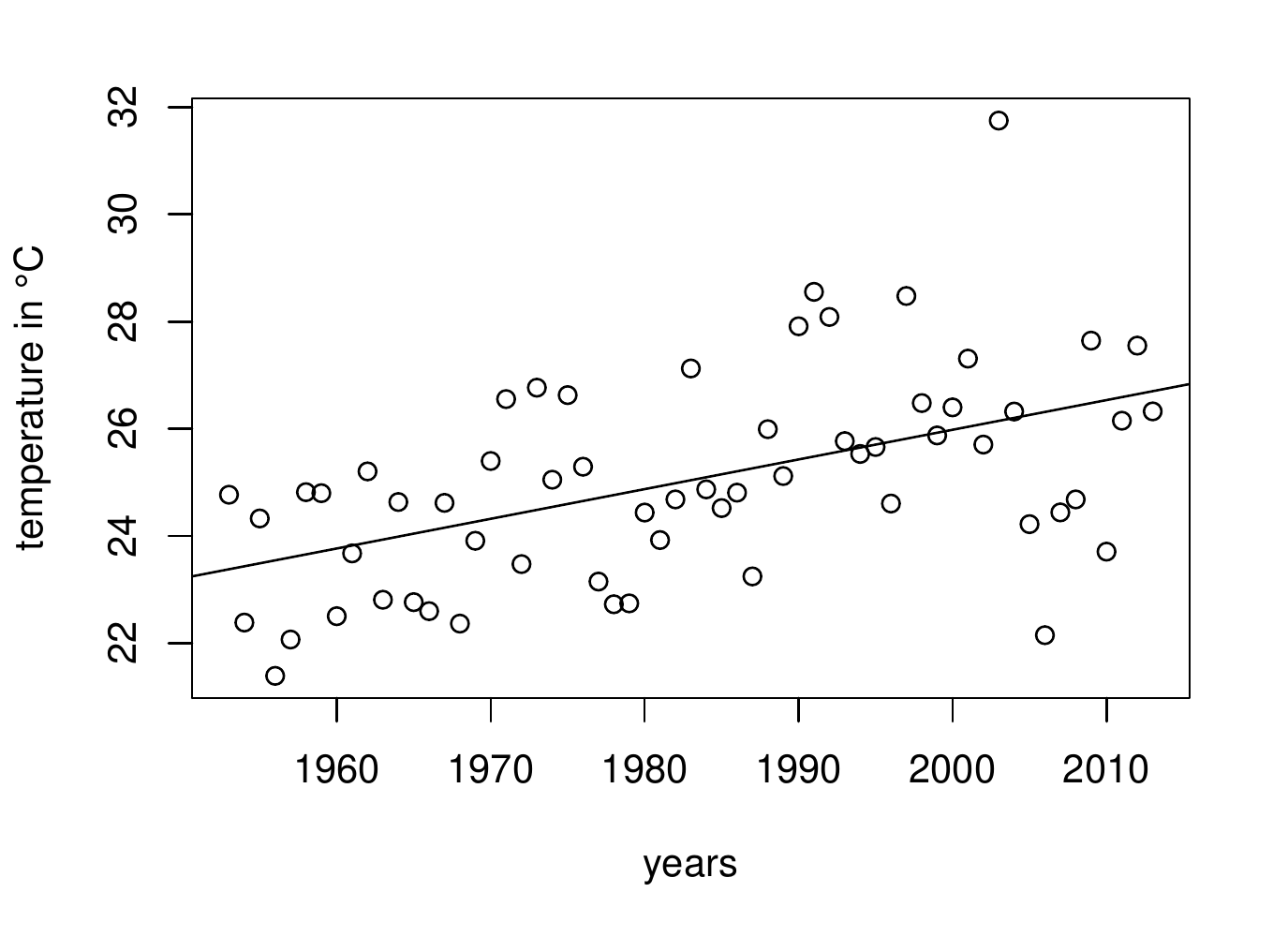}%
\caption{Monthly averages of maximum daily temperature in the month of August, 1953 to
2013, in Mannheim, Germany. Note the very high record of almost
$32^\circ\text{C}$ during the European heat wave in Summer 2003.}%
\label{fig:august}%
\end{figure}

The El Ni\~{n}o indices are already smoothed regional averages and do not
require further smoothing. We therefore expand the monthly
measurements using a B-spline basis of order four with one knot placed
at each month, the expansion closely matches the data points.
Figure~\ref{f:ninoregions} shows the spline-expansion for the year
2012, while Figure~\ref{f:ninofullsample} exhibits the SST
measurements in each region over the whole span of the sample.  The
data from individual weather stations in the US and in Germany are
noisy. Due to pronounced annual periodicity, we expand each annual
curve using a Fourier series with 25 basis functions and apply the
harmonic acceleration roughness penalty. An example for the US smoothed,
log-transformed precipitation data is given in Figure \ref{f:SantaCruz}.  
The smoothness parameter was chosen to minimize the standard 
generalized cross--validation criterion. Details of the procedure 
are described in Section 5.3 of \citet{ramsay:hooker:graves:2009}. 

We apply the test to the detrended and smoothed curves.  As described
in Section~\ref{ss:fsi}, Step 3, the numbers $p(i)$ are determined
separately for each time series $i$ such that at least 85\% of
variance within the time series is captured by the $p(i)$ principal
components. For the El Ni\~{n}o SST curves, this requires two
principal components for each region.  For the log-precipitation and
the temperature data, the same criterion selects three principal
components for each station.  The fact that all time series within a
panel are assigned the same number of principal components underlines
the fact that the climate measures we use have structurally similar
data generating processes. This, however, does not necessarily need to
be the case for different panels. Our testing procedure accounts
 for such a possibility by
choosing $p(i)$ individually for each $i$. Additional simulations show
that the test is very robust to the choice of the cut-off criterion of
this first dimension reduction.  
The second principal component analysis described in Step 5,
Section~\ref{ss:fsi}, determines the number of dimensions $q$ for each
sample that are finally used to derive our test statistic. The 85\% of
variance criterion selects $q=2$ for the El Ni\~{n}o data with
$I=4$. It selects $q=3$ for the German temperature data with $I=42$,
and $q=37$ for the US precipitation data with $I=103$. The fact that
the second dimension reduction for the German temperature data selects
just 3 principal components for the whole panel can be attributed to
the rather dense geographical coverage and thus large homogeneity of
the cross-sectional temperature curves in this sample.

\begin{figure}[htb]
\centering
\includegraphics[width=\columnwidth]{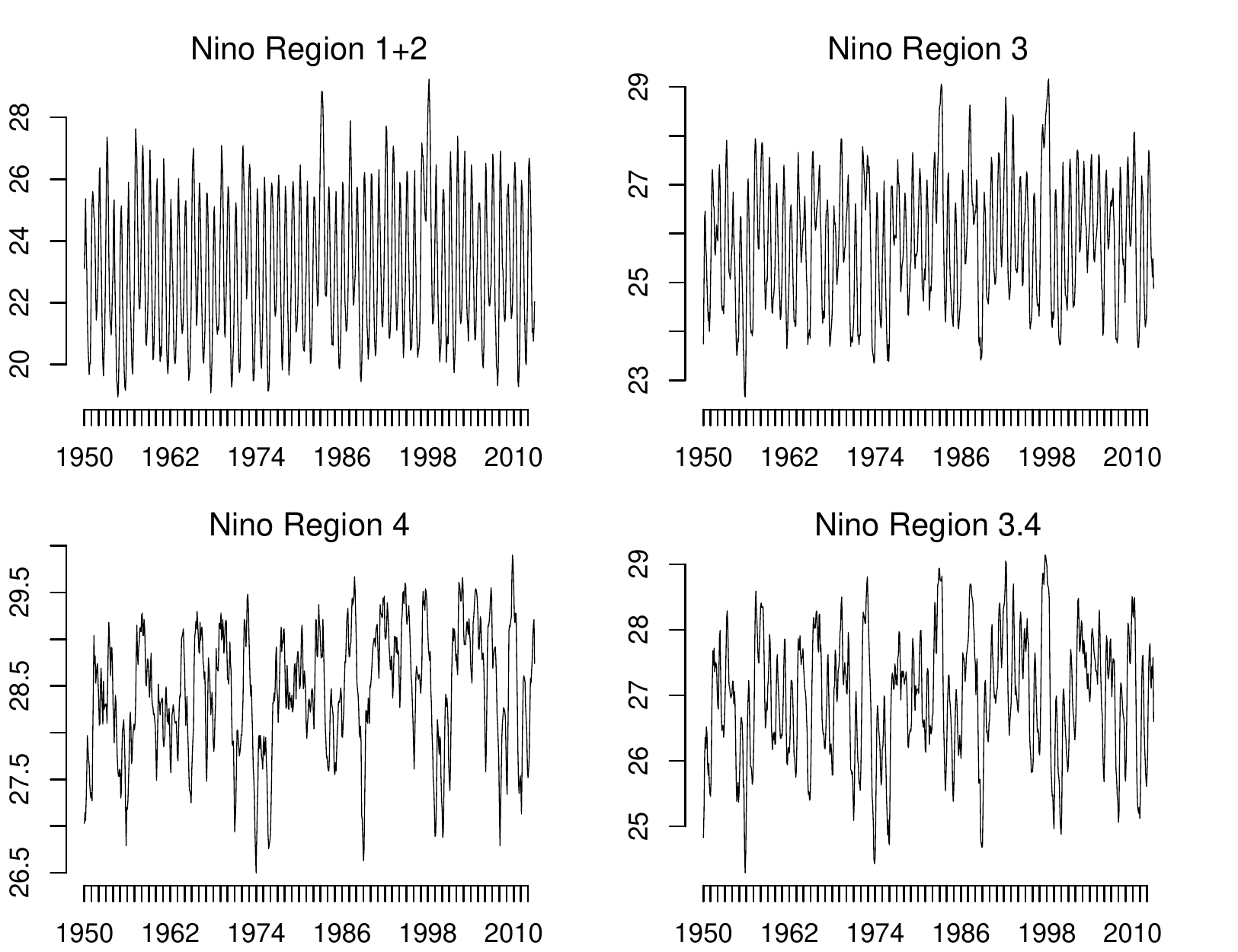}
\caption{Sea Surface Temperature in four El Ni\~{n}o regions
 from 1950 to 2012}%
	\label{f:ninofullsample}%
\end{figure}

\begin{table}[htb]
\centering
\caption{Test results (normalized $\widehat{Q}_N$ and $p$\,-values)
for the three data sets:
El Ni\~{n}o SST curves with $I=4$, $N=63$;
US log-precipitation data with $I=103$, $N=113$;
German temperature data with $I=42$, $N=61$.}
\label{t:testresults}
\begin{tabular}{l|rr|rr|rr}
\hline\hline
& \multicolumn{2}{c|}{El Ni\~{n}o SST}
& \multicolumn{2}{c|}{US precip.}
& \multicolumn{2}{c}{German temp.} \\
& stat. & $p$\,-value & stat. & $p$\,-value & stat. & $p$\,-value\\
\hline
$H=1$	&	13.483	&	$<0.001$	&	1.704	&	0.044	&	-0.701	&	0.758	\\
$H=2$	&	11.778	&	$<0.001$	&	2.080	&	0.019	&	-1.000	&	0.841	\\
$H=3$	&	10.282	&	$<0.001$	&	2.827	&	0.002	&	-0.138	&	0.555	\\
$H=4$	&	 8.973	&	$<0.001$	&	2.587	&	0.005	&	 0.193	&	0.424	\\
$H=5$	&	 7.617	&	$<0.001$	&	2.714	&	0.003	&	-0.349	&	0.636	\\
$H=6$	&	 6.508	&	$<0.001$	&	2.762	&	0.003	&	-0.119	&	0.547	\\
$H=7$	&	 5.756	&	$<0.001$	&	2.696	&	0.004	&	 0.066	&	0.474	\\
$H=8$	&	 5.142	&	$<0.001$	&	3.060	&	0.001	&	 0.324	&	0.373	\\
$H=9$	&	 4.883	&	$<0.001$	&	3.008	&	0.001	&	 0.011	&	0.496	\\
$H=10$&	 4.750  &	$<0.001$	&	2.871	&	0.002	&	 0.231	&	0.408	\\
\hline\hline
\end{tabular}
\end{table}

Table~\ref{t:testresults} shows the normalized test statistic and
the corresponding $p$\,-values for $H=1,...,10$ for all three panels.
The El Ni\~{n}o anomaly typically happens at irregular intervals of 
three to six, sometimes seven years, thus, importance should be
given to the results for these lags.
However, the test rejects $H_0$ for the SST panel for all lags tested
at any reasonable level of significance.  
For the South--West precipitation data, the rejection is also 
convincing, even though less strong.  Both data sets reflect
known semi--periodic cycles extending over several years, and this is
a  likely reason for the rejections.  For the German temperature data,
in contrast, there is no evidence for a violation of
$H_0$. After a simple detrending, the panel of annual temperature
curves over Germany can be taken to consist of iid observations.  Due
to a dense spatial coverage, one might say that after local trends
have been removed, the annual temperature pattern over Germany can be
treated each year as an independent replication.

\subsection{Simulation scheme and finite sample performance}
\label{ss:sim}
We now use the estimated stochastic structure of the  panels
described in Section~\ref{ss:data} to generate artificial functional
panels.  This section serves a twofold purpose: we want to validate
the conclusions implied by Table \ref{t:testresults},  and we want to
evaluate the empirical size and power of the test in realistic settings.

\subsubsection{Simulation scheme}
We first describe the procedure to simulate functional panels which
closely resemble the original climate data, but do not violate the null
hypothesis. Then, we explain how we generate panels with increasing
temporal dependence.
Functional panel data are expected to exhibit correlation of curves 
for a fixed period $n$ (`between' individuals, i.e., stations or regions), 
as well as dependence over periods for a fixed individual $i$ 
(`within' a time series).  
The following data generating process features the first type of between 
correlation but excludes dependence over time:
For every individual station or region $i$, we calculate the
empirical mean function $\hat{\mu}_i(t)$ together with $k=1,...,12$
empirical principal components $\hat{v}_{k,i}(t)$ (12 is the maximum
number of EFPC's for these data).  We calculate the score
$\hat{\xi}_{k,i,n}$ for each principal component $k$, for every
individual $i$, and for every year $n=1,...,N$.  Let
\[
\sigma_k(i, \ip) = \frac{1}{N-1} \sum_{n=1}^N
\lp \hat\xi_{k,i,n} - \bar \xi_{k,i} \rp
\lp \hat\xi_{k,\ip,n} - \bar \xi_{k,\ip} \rp
\]
and
\[
\bSig_k = \lb \sigma_k(i, \ip), 1 \le i, \ip \le I \rb.
\]
The matrix $\bSig_k$ is the empirical covariance matrix
of the scores from different individuals. Calculate
its Cholesky decomposition $\bSig_k = \bL_k \bL_k^\top$ and
obtain simulated scores of the form
\[
\boldsymbol{\zeta}_k =  \mathbf{z}_k \bL^\top_k, \quad 1 \le k \le 12,
\]
where $\mathbf{z}_k$ is an $N \times I$ matrix of independent
standard normal random variables. Each matrix
$\boldsymbol{\zeta}_k$, $1\leq k\leq 12$,
has $I$ column vectors $\boldsymbol{\zeta}_{k,i}$
of length $N$ which are correlated among each other much like
the scores of the individual stations or regions from the original
climate data. With these scores, the data generating process
for the simulated iid functional panel is
\[
\mathbf{X}^{H_0}_{i}(t) = \hat{\mu}_i(t) + \sum_{k=1}^{12} \boldsymbol{\zeta}_{k,i} \hat{v}_{k,i}(t), \quad i=1,...,I,
\]
where each $\mathbf{X}^{H_0}_{i}(t)$ is a vector of random curves
of length $N$ and the superscript $H_0$ indicates that the artificial
sample satisfies the null hypothesis of independence. The essence of
the above procedure is that if the original data satisfied $H_0$ and
were normal, then the data generating process would have the same random
structure as the estimated structure of the data. Normal QQ--plots show
that the scores of the three data sets in question are approximately
normal ($i, k$ fixed, $N$  points per plot).

To construct an  alternative to $H_0$, we impose autocorrelation on
each time series in the form of a functional autoregressive process of
order 1,  FAR(1) (see Chapters 3 and 4 of
\citet{bosq:2000}, or Chapter 13 of \citet{HKbook}).
An appropriate  Cholesky factor $\mathbf{L}_{\text{ac}}$ is defined
as follows:
Choose $\rho\neq0$, $-1<\rho<1$, the level of autocorrelation to
be imposed (one could also specify different levels of $\rho$ for
each $k$). Construct an $N \times N$ Toeplitz matrix
such that the first column corresponds to the sequence
$\{\rho^{n-1}\}_{n=1,...,N}$.
For the Cholesky factor $\mathbf{L}_{\text{ac}}$, take the lower
triangular of this Toeplitz matrix and divide each element by
$[(\rho^{2n}-1)/(\rho^2-1)]^{1/2}$, where $n=1,...,N$
denotes the row number of the corresponding element. This ensures
that $\mathbf{L}_{\text{ac}} \mathbf{L}_{\text{ac}}^\top$ is positive
semi-definite and has all diagonal elements equal to $1$. Thus,
applying this factor to a vector of length $N$ imposes autocorrelation
among the vectors' elements but does not change the
overall level of variance. With the generic scores defined before,
the data generating process for the autocorrelated functional
panel is
\[
\mathbf{X}^{\text{ac}}_{i}(t) = \hat{\mu}_i(t)
				+ \sum_{k=1}^{12} \mathbf{L}_{\text{ac}}\boldsymbol{\zeta}_{k,i} \hat{v}_{k,i}(t)
				, \quad i=1,...,I.
\]

In summary, we can generate siblings of functional panels of the type
$\{X^{H_0}_{i,n}(t), X^{\text{ac}}_{i,n}(t) \}$,
$i=1,...,I$, $n=1,...,N$, whose cross--sectional dependence structure
is the same and similar to that of the real data, but where one sample
obeys the hypothesis of independence while the other sample follows an
explicit FAR(1) process. This procedure allows us to vary the length of
the panel, $N$.

\subsubsection*{Finite sample performance}
To evaluate the empirical size and power, we simulate $R=10^3$
replications of panels with $I=4$, $I=42$, and $I=103$, and with the
cross--sectional dependence structure resembling that of the
respective data sets. We report results for the length $N=60$ and
$N=120$ (years), typical sample sizes encountered in historical
climate data.  We test for $H=3,...,6$ which is the relevant range of
years for which dependence in the El Ni\~no driven climatic measures
is expected.  Table~\ref{t:empiricalsize} shows the point estimates of
the rejection frequency for the simulation under $H_0$ together with
the Clopper--Pearson confidence intervals for the probability of
success (see \citet{clopper:pearson:1934}).\footnote{ Clopper--Pearson
confidence intervals are almost identical to the confidence intervals
based on the normal approximation to the binomial distribution, except
for cases of empirical power close to 1 when the right end point of
the latter exceeds 1.}

\begin{table}[hb]
\centering
\caption{Rejection frequencies (and confidence bands) for the test
with $\alpha=0.05$ obtained from 1000 simulations under $H_0$ for
each of the data sets.}
\label{t:empiricalsize}
\begin{tabular}{l|ll|ll}
\multicolumn{5}{c}{El Ni\~{n}o data, $I=4$, i.i.d.}\\
  \hline\hline
& \multicolumn{2}{l|}{$N=60$}   & \multicolumn{2}{l}{$N=120$}  \\
  \hline
  $H=3$ & 0.055 & (0.042, 0.071) & 0.053 & (0.040, 0.069) \\
  $H=4$ & 0.071 & (0.056, 0.089) & 0.051 & (0.038, 0.067) \\
  $H=5$ & 0.060 & (0.046, 0.077) & 0.061 & (0.047, 0.078) \\
  $H=6$ & 0.063 & (0.049, 0.080) & 0.057 & (0.043, 0.073) \\
  \hline
\multicolumn{5}{c}{} \\
\multicolumn{5}{c}{US precipitation data, $I=103$, i.i.d.}\\
  \hline\hline
& \multicolumn{2}{l|}{$N=60$}   & \multicolumn{2}{l}{$N=120$}  \\
  \hline
  $H=3$ & 0.033 & (0.023, 0.046) & 0.033 & (0.023, 0.046) \\
  $H=4$ & 0.046 & (0.034, 0.061) & 0.041 & (0.030, 0.055) \\
  $H=5$ & 0.054 & (0.041, 0.070) & 0.048 & (0.036, 0.063) \\
  $H=6$ & 0.079 & (0.063, 0.097) & 0.058 & (0.044, 0.074) \\
  \hline
\multicolumn{5}{c}{} \\
\multicolumn{5}{c}{German temperature data, $I=42$, i.i.d.}\\
  \hline\hline
& \multicolumn{2}{l|}{$N=60$}   & \multicolumn{2}{l}{$N=120$}  \\
  \hline
	$H=3$ & 0.062 & (0.048, 0.079) & 0.061 & (0.047, 0.078) \\
  $H=4$ & 0.060 & (0.046, 0.077) & 0.052 & (0.039, 0.068) \\
  $H=5$ & 0.063 & (0.049, 0.080) & 0.054 & (0.041, 0.070) \\
  $H=6$ & 0.059 & (0.045, 0.075) & 0.054 & (0.041, 0.070) \\
  \hline\hline
\end{tabular}
\end{table}

The empirical sizes reported in Table~\ref{t:empiricalsize} validate
the results obtained in Section~\ref{ss:data}.  For the DGPs we
considered, the test has overall a satisfactory, often excellent,
empirical size.  The evaluation of power is more subjective as it
depends on the distance of the DGP from $H_0$. For every
simulated panel following $H_0$, we obtained an autocorrelated sibling
with a fixed level of autocorrelation: $\rho=0.38$ for $I=4$,
$\rho=0.37$ for $I=42$, and $\rho=0.19$ for $I=103$.  These are the
correlation levels for which the power is almost or exactly equal to
1 if $N=120$. In light of these moderate levels of autocorrelation
and the rejection frequencies reported in Table~\ref{t:powertable}, we
conclude that our test has excellent power together with good
empirical size, such that the rejections as well as the non--rejection
of $H_0$ for the data presented in Section~\ref{ss:data} provide
reliable insights.

\begin{table}[htb]
\centering
\caption{Rejection frequencies (and confidence bands) for the test with $\alpha=0.05$ obtained from 1000 simulations of an AR(1) process for each of the data sets; $\rho$ indicates the degree of autocorrelation.}
\label{t:powertable}
\begin{tabular}{l|ll|ll}
\multicolumn{5}{c}{El Ni\~{n}o SST curves, $I=4$, $\rho=0.38$}\\
  \hline\hline
& \multicolumn{2}{l|}{$N=60$}   & \multicolumn{2}{l}{$N=120$}  \\
  \hline
	 $H=3$ & 0.987 & (0.978, 0.993) & 1.000 & (0.996, 1.000) \\
  $H=4$ & 0.927 & (0.909, 0.942) & 1.000 & (0.996, 1.000) \\
  $H=5$ & 0.778 & (0.751, 0.803) & 1.000 & (0.996, 1.000) \\
  $H=6$ & 0.607 & (0.576, 0.637) & 0.997 & (0.991, 0.999) \\
	\hline
\multicolumn{5}{c}{} \\
\multicolumn{5}{c}{US precipitation data, $I=103$, $\rho=0.19$}\\
  \hline\hline
& \multicolumn{2}{l|}{$N=60$}   & \multicolumn{2}{l}{$N=120$}  \\
  \hline
  $H=3$ & 0.951 & (0.936, 0.964) & 1.000 & (0.996, 1.000) \\
  $H=4$ & 0.981 & (0.970, 0.989) & 1.000 & (0.996, 1.000) \\
  $H=5$ & 0.994 & (0.987, 0.998) & 1.000 & (0.996, 1.000) \\
  $H=6$ & 0.994 & (0.987, 0.998) & 1.000 & (0.996, 1.000) \\
\hline
\multicolumn{5}{c}{} \\
\multicolumn{5}{c}{German temperature curves, $I=42$, $\rho=0.37$}\\
  \hline\hline
& \multicolumn{2}{l|}{$N=60$}   & \multicolumn{2}{l}{$N=120$}  \\
  $H=3$ & 0.790 & (0.763, 0.815) & 0.996 & (0.990, 0.999) \\
  $H=4$ & 0.690 & (0.660, 0.719) & 0.986 & (0.977, 0.992) \\
  $H=5$ & 0.615 & (0.584, 0.645) & 0.974 & (0.962, 0.983) \\
  $H=6$ & 0.554 & (0.523, 0.585) & 0.966 & (0.953, 0.976) \\
  \hline\hline
\end{tabular}
\end{table}

In our simulation study, the power generally decreases with $H$, but
it increases with $H$ for the DGP mimicking the US precipitation
panel. This agrees with the $p$\,-values reported in
Table~\ref{t:testresults}.  The issue of the selection of $H$ is a
difficult one, and it is not satisfactorily solved even for the
standard Ljung--Box--Pierce test for a single scalar time series.
Statistical software packages display the $p$\,-values, often in the
form of a graph, as a function of $H$. If for some relatively large
range of $H$ the $p$\,-values are above the 5\% level, $H_0$ is
accepted, if they are below, $H_0$ is rejected. In mixed cases, the
test is found to be inconclusive.  The same strategy can be followed
in the application of our test.  In addition, some background
knowledge of the science problem may be utilized. In the SST and US
precipitation examples, the temperature and rainfall patterns are
known to reoccur every 3--6 years, so importance was attached to these
lags $H$.


\appendix

\section{Proof of Theorems~\ref{t:multi} and \ref{t:multi:HA}} \label{s:multi}
The plan of the proof of Theorem~\ref{t:multi} is as follows.
In Lemma~\ref{l:2}, we show that the convergence to the
normal limit holds if instead of projections on the
EFPCs $\hv_{i,j}$, projections on the $v_{i,j}$ are used.
Recall that $v_{i,j}$ is the $j^\text{th}$ FPC of the $i^\text{th}$ functional
time series in the panel.
The statistic constructed using the $v_{i,j}$ is denoted by
$Q_N$.
The proof of Lemma~\ref{l:2} relies on Theorem~\ref{t:norm}.
Next, we show in Lemmas \ref{l:3} and \ref{l:4} that the
the transition from $Q_N$ to $\widehat Q_N$ involves
asymptotically negligible terms. Lemma~\ref{l:props}
collects several properties referred to  in the proofs.  {The proof of Theorem \ref{t:multi:HA} is given at the end of this section.}

Let $\bX_{i,n}$ be  the  column vector of length $p(i)$ defined by
$
\bX_{i,n} = [\langle X_{i,n}, v_{i,j} \rangle, \ 1 \le j \le p(i)]^\top.
$
By stacking these $I$ vectors on top of each other, we
construct a column vector of length $p= \sum_{i=1}^I p(i)$
defined by
\[
\bX_n = [ \bX_{1,n}^\top, \ldots, \bX_{I,n}^\top]^\top.
\]
We abuse notation slightly as $\bX_n$ was used earlier to reference the functional panel vector, but throughout this section $\bX_n$ will be defined as above.  We allow the number of time series, $I$, and/or the number of
FPCs, $p(i)$, used for each series to increase with the temporal sample
size, $N$, in any way which implies that $p=p_N$ increases to
infinity. Recall that
$
\bC_{0,N} = \E \lb\bX_n \bX_n^\top \rb
$
is the $p_N\times p_N$ covariance matrix whose eigenvalues are
$\gamma_1 \ge \ldots \ge \gamma_{p_N}$.
Our first lemma contains two bounds involving the $\gamma_j$,
which will be used throughout the proof of Theorem~\ref{t:multi}.

\begin{lemma}\label{l:bound}  We have the bounds
\[
\sum_{j=1}^{p_N} \gamma_j \leq I M^{1/2}
\ \ \  \mbox{\rm and} \ \ \
p_N \leq \gamma_{p_N}^{-1} I M^{1/2},
\]
where $M$ is the bound in Assumption~\ref{a:1}.
\begin{proof}
Notice that
\[
\sum_{j=1}^{p_N} \gamma_j = \mbox{trace}(\bC_{0,N}) = \E |\bX_n|^2
 \leq \sum_{i=1}^I \E \|X_{i,n}\|^2.
\]
Applying Jensen's inequality and Assumption \ref{a:1} gives the first claim.
We then immediately obtain the second claim since
$
\gamma_{p_N} p_N \leq \sum_{j=1}^{p_N} \gamma_j.
$
\end{proof}
\end{lemma}

\begin{lemma}  \label{l:2}
If Assumptions \ref{a:1} and \ref{a:pn2} hold then,
 under $H_0$,
\[
\frac{Q_N - p_N^2 H}{p_N \sqrt{2H}} \overset{\cD}{\to} \mathcal{N}(0,1).
\]
\end{lemma}
\begin{proof}
The standardized vectors used in Theorem~\ref{t:norm} are given by
\[
\bZ_{n,N} = \bC^{-1/2}_{0,N} \bX_{n}.
\]
Using the mixed--product property of the Kronecker product, we can express $Q_N$ as
\begin{align*}
& N^{-2} \sum_{n=1}^{N-h} \sum_{\np=1}^{N-h} (\bX_n \otimes \bX_{n+h}) ^\top(\bC_{0,N} \otimes \bC_{0,N})^{-1} (\bX_n \otimes \bX_{n+h}) \\
= & N^{-2} \sum_{n=1}^{N-h} \sum_{\np=1}^{N-h} (\bX_n \otimes \bX_{n+h}) ^\top(\bC_{0,N}^{-1/2} \otimes \bC_{0,N}^{-1/2})
(\bC_{0,N}^{-1/2} \otimes \bC_{0,N}^{-1/2})
(\bX_n \otimes \bX_{n+h}) \\
= & N^{-2} \sum_{n=1}^{N-h} \sum_{\np=1}^{N-h} (\bZ_{n,N} \otimes \bZ_{n+h,N})^\top (\bZ_{n,N} \otimes \bZ_{n+h,N}) \\
= & N^{-2} \sum_{n=1}^{N-h} \sum_{\np=1}^{N-h} \left| \bZ_{n,N} \otimes \bZ_{n+h,N} \right|^2.
\end{align*}
So we need to establish that
\[
N^{-1/2} \E|\bZ_{1,N}|^4 \to 0.
\]
By definition we have that
\[
|\bZ_{1,N}|^4 =  (\bX_{1}^\top \bC^{-1}_{0,N} \bX_{1})^2.
\]
Applying the Cauchy--Schwarz and operator norm inequality we have that
\[
\bX_{1}^\top \bC^{-1}_{0,N} \bX_{1} \leq | \bX_{1} |  |\bC^{-1}_{0,N}
 \bX_{1}| \leq | \bX_1|^2 \| \bC^{-1}_{0,N}  \|.
\]
Since $\|\bC^{-1}_{0,N} \|$ is the largest eigenvalue of
$\bC^{-1}_{0,N} $, it is simply the reciprocal of the smallest
eigenvalue of $\bC_{0,N} $.  Therefore we have
\[
|\bZ_{1,N}|^4 \leq | \bX_{1}|^4 \gamma_{p_N}^{-2}.
\]
The norm of $\bX_{1}$ can be expressed as
\[
|\bX_{1}|^2 = \sum_{i}^I \sum_{j}^{p(i)} \langle X_{i,1}, v_{j,i} \rangle^2
\leq \sum_{i}^I \|X_{i,n}\|^2.
\]
A final application of the Cauchy--Schwarz inequality yields
\[
|\bZ_{1,N}|^4
\leq \gamma_{p_N}^{-2} I \sum_{i=1}^I \|X_{i,1}\|^4.
\]
Taking expected values we have that
\[
\E|\bZ_{1,N}|^4 \leq  \gamma_{p_N}^{-2} I^2 M.
\]
Therefore by Assumption \ref{a:pn2}
\[
N^{-1/2} \E|\bZ_{1,N}|^4 = N^{-1/2}  \gamma_{p_N}^{-2} I^2 =   o(1).
\]
Finally, to apply Theorem \ref{t:norm} we need only to show that $p_N N^{-2/3} \to 0$.  Using Lemma \ref{l:bound} we have that
\[
p_N N^{-2/3}\leq N^{-2/3} \gamma_{p_N}^{-1} I M^{1/2}.
\]
Which is clearly $o(1)$ by Assumption \ref{a:pn2}.
\end{proof}

For the next lemma it will be notationally useful to define the lag $h$ cross covariance operators:
\[
\Delta_{i,\ip,h} = N^{-1/2} \sum_{n=1}^{N - h} X_{i,n} \otimes X_{\ip, n+h}.
\]

\begin{lemma}  \label{l:3}
If Assumptions \ref{a:1} and \ref{a:pn2} hold then,
 under $H_0$,
\[
\frac{Q_N - N \sum_{h=1}^H \widehat V_h^\top
(\bC_{0,N} \otimes \bC_{0,N})^{-1} \widehat V_h}{p_N \sqrt{2H}}
= o_P(1).
\]
\end{lemma}
\begin{proof}
A minor rearrangement yields
\[
\frac{Q_N - N \sum_{h}^H \widehat V_h^\top (\bC_{0,N} \otimes \bC_{0,N})^{-1} \widehat V_h}{p_N \sqrt{2H}}
 = \frac{2 N \sum_{h}^H ( V_h - \widehat V_h )^\top (\bC_{0,N} \otimes \bC_{0,N})^{-1} (V_h+\widehat V_h)}{p_N \sqrt{2H}}.
\]
The Cauchy--Schwarz and operator norm inequality yield
\begin{align*}
  |( V_h - \widehat V_h )^\top (\bC_{0,N} \otimes \bC_{0,N})^{-1} (V_h+\widehat V_h)|
  \leq \|V_h - \widehat V_h\| \|V_h + \widehat V_h\| \gamma_{p_N}^{-2}
\end{align*}
For each coordinate of $V_h$, there exists $i,j, \ip, \jp$ such that the coordinate can be expressed as
 \[
  N^{-1} \sum_{n=1}^{N-h} \langle X_{i,n}, v_{j,i} \rangle  \langle X_{i^\prime,n+h}, v_{j^\prime,i^\prime} \rangle
  =  N^{-1} \sum_{n=1}^{N-h} \langle X_{i,n} \otimes X_{i^\prime,n+h}, v_{j,i}\otimes v_{j^\prime,i^\prime} \rangle
  = N^{-1/2} \langle \Delta_{i,\ip,h} , v_{j,i}\otimes v_{j^\prime,i^\prime} \rangle.
    \]
 Therefore
 \begin{align*}
 \|V_h - \widehat V_h\|^2
& =  N^{-1} \sum_{i, \ip}^I \sum_{j}^{p(i)} \sum_{\jp}^{p(\ip)}  \langle \Delta_{i,\ip,h} , v_{j,i}\otimes v_{j^\prime,i^\prime}
- \hat v_{j,i}\otimes \hat v_{j^\prime,i^\prime} \rangle^2 \\
& \leq N^{-1} \max \|\Delta_{i,\ip,h}\|^2  \sum_{i, \ip}^I \sum_{j}^{p(i)} \sum_{\jp}^{p(\ip)} \| v_{j,i}\otimes v_{j^\prime,i^\prime}
- \hat v_{j,i}\otimes \hat v_{j^\prime,i^\prime}\|^2 \\
& \leq N^{-1} \max \|\Delta_{i,\ip,h}\|^2  \sum_{i, \ip}^I \sum_{j}^{p(i)} \sum_{\jp}^{p(\ip)} ( \alpha_{i,j} \|C_{ii} - \widehat C_{ii}\| +
\alpha_{\ip,\jp} \|C_{\ip \ip} - \widehat C_{\ip \ip}\| )|^2 \\
& \leq N^{-1} \max \|\Delta_{i,\ip,h}\|^2 \max \|C_{ii} - \widehat C_{ii}\|^2 \Gamma_{N}\\
& = N^{-2} I^3 \Gamma_{N} O_P(1),
\end{align*}
where the last equality follows from Lemma \ref{l:props}.  By Parceval's inequality
\[
\| V_h \|^2 =
N^{-1} \sum_{i, \ip}^I \sum_{j}^{p(i)} \sum_{\jp}^{p(\ip)}  \langle \Delta_{i,\ip,h} , v_{j,i}\otimes v_{j^\prime,i^\prime} \rangle^2
\leq N^{-1} \sum_{i, \ip} \|\Delta_{i,\ip,h}\|^2
 = N^{-1} O_P(I^2),
\]
and the same holds for $\| \widehat V_h\|^2$.
Combining everything, the original difference is of the order
\[
N p_N^{-1} \gamma_{p_N}^{-2} N^{-1} I^{3/2} \Gamma_N^{1/2}N^{-1/2} I
O_P(1) = N^{-1/2} p_N^{-1} I^{5/2} \gamma_{p_N}^{-2}
 \Gamma_N^{1/2}O_P(1),
\]
which is $o_P(1)$ by Assumption \ref{a:pn2}.
\end{proof}

\begin{lemma}  \label{l:4}
 If Assumptions \ref{a:1} and \ref{a:pn2} hold,
then under $H_0$,
\[
 \frac{ N \sum_{h=1}^H
[\widehat V_h^\top (\bC_{0,N} \otimes \bC_{0,N})^{-1} \widehat V_h -
\widehat V_h^\top (\widehat \bC_{0,N} \otimes \widehat \bC_{0,N})^{-1}
\widehat V_h]}{p_N \sqrt{2H}} =o_P(1).
\]
\end{lemma}
\begin{proof}
Using the Cauchy--Schwarz inequality and
operator norm inequality we have that
\begin{align*}
& |\widehat V_h^\top (\bC_{0,N} \otimes \bC_{0,N})^{-1} \widehat V_h -
\widehat V_h^\top (\widehat \bC_{0,N} \otimes \widehat \bC_{0,N})^{-1} \widehat V_h |\\
\leq & \| \widehat V_h^\top\|^2 \|(\bC_{0,N} \otimes \bC_{0,N})^{-1} -  (\widehat \bC_{0,N} \otimes \widehat \bC_{0,N})^{-1} \|_{\cS} \\
\leq &  \| \widehat V_h\|^2 [\gamma_{p_N}^{-1} +  \widehat \gamma_{p_N}^{-1} ] \| \bC_{0,N}^{-1} -  \widehat \bC_{0,N}^{-1} \|_{\cS} \\
\leq &  \| \widehat V_h\|^2 [\gamma_{p_N}^{-1} +  \widehat \gamma_{p_N}^{-1} ] [ \gamma_{p_N}^{-1}  \widehat \gamma_{p_N}^{-1} ]\| \bC_{0,N} -  \widehat \bC_{0,N} \|_{\cS} \\
\leq &  \| \widehat V_h\|^2 [\gamma_{p_N}^{-1} +  \widehat \gamma_{p_N}^{-1} ] [ \gamma_{p_N}^{-1}  \hat \gamma_{p_N}^{-1} ]\| \bC_{0,N} -  \hat \bC_{0,N} \|.
\end{align*}
As before, we can apply Lemma \ref{l:props}.1 to obtain
\[
\| \widehat V_h\|^2 = \widehat V_h^\top \widehat V_h
= N^{-1} \sum_{i, \ip}^I \sum_{j}^{p(i)} \sum_{\jp}^{p(\ip)}  \langle \Delta_{i,\ip,h} , \hat v_{j,i}\otimes \hat v_{j^\prime,i^\prime} \rangle^2
\leq N^{-1} \sum_{i, \ip}^I \|  \Delta_{i,\ip,h}\|^2
 = O_P(N^{-1} I^2).
\]
Turning to the eigenvalues we have that
\[
[\gamma_{p_N}^{-1} +  \hat \gamma_{p_N}^{-1} ] [ \gamma_{p_N}^{-1}  \hat \gamma_{p_N}^{-1} ]
 = \gamma_{p_N}^{-3} \left[ \frac{\gamma_{p_N} }{\hat \gamma_{p_N}} +  \frac{\gamma_{p_N}^2 }{\hat \gamma_{p_N}^2}\right].
\]
Applying Lemma \ref{l:props}.4 to bound the difference
\[
\left|\frac{\hat \gamma_{p_N} }{ \gamma_{p_N}} - 1\right|
= \frac{|\hat \gamma_{p_N} - \gamma_{p_N}| }{ \gamma_{p_N}}
\leq \frac{ \| \widehat \bC_{0,N} - \bC_{0,N}\| }{\gamma_{p_N}}
= O_P(\gamma_{p_N}^{-1} I N^{-1/2} \Gamma_N^{1/2}) =o_P(1).
\]
Putting everything together, the original difference is
\[
O_P(\gamma_{p_N}^{-3}  I^3 N^{-1/2} p_N^{-1} \Gamma_N^{1/2} ) = o_P(1),
\]
by Assumption \ref{a:pn2}.
\end{proof}

Our last lemma contains several properties which were
used in the arguments developed above.

\begin{lemma} \label{l:props}   If Assumptions \ref{a:1} and \ref{a:pn2}
 hold, then we have the following properties:
\begin{enumerate}
\item $\max \| \Delta_{i,\ip,h}\|^2 = O_P(I^2)$ under $H_0$.
\item $\max \| C_{i \ip}\| = O(1)$.
\item $\max \| \widehat C_{ii} - C_{ii} \| = O_P(N^{-1/2} I )$ and $\max \| \widehat C_{ii} - C_{ii} \|^2 = O_P(N^{-1} I )$

\item $\|  \widehat \bC_{0,N} - \bC_{0,N} \|
= O_P(I N^{-1/2} \Gamma_N^{1/2})$ \\
\end{enumerate}

\begin{proof}
\begin{enumerate}
\item  For each fixed $i$ and $\ip$, we have that
\begin{align*}
 \E\| \Delta_{i,\ip,h}\|^2 & = N^{-1} \sum_{n=1}^{N-h} \sum_{\np=1}^{N-h} \langle X_{i,n} \otimes X_{i,n+h}, X_{\ip,\np} \otimes X_{\ip,\np+h} \rangle \\
& = N^{-1} \sum_{n=1}^{N-h} \E \langle X_{i,n} \otimes X_{i,n+h}, X_{\ip,n} \otimes X_{\ip,n+h} \rangle \\
& \leq \E\|X_{i,1}\|^2 \E\|X_{\ip,1}\|^2 \leq M.
\end{align*}
Therefore we have that
\[
\E \left[\max \| \Delta_{i,\ip,h}\|^2\right] \leq I^2 M,
\]
and the result follows from Markov's inequality.

\item By Jensen's inequality we have that
\begin{align*}
 \|C_{i \ip}\| &\leq  \E\|X_{i,n}\otimes X_{\ip,n}\| = \E[\|X_{i,n}\| \|X_{\ip,n}\|] \leq \sqrt{\E\|X_{i,n}\|^2 \E\|X_{\ip,n}\|^2} \\
& \leq (\E\|X_{i,n}\|^4 \E\|X_{\ip,n}\|^4)^{1/4} \leq M^{1/2},
\end{align*}
which proves the claim.

\item  The argument is the same as in 1.

\item By the triangle inequality we have
\[
\|  \widehat \bC_{0,N} - \bC_{0,N} \|
\leq \|  \widehat \bC_{0,N} -  \tilde\bC_{0,N}  \|
+  \| \tilde  \bC_{0,N} -  \bC_{0,N}\|
\]
where $ \tilde\bC_{0,N}$ is formed by projecting the $C_{i\ip}$
onto the estimated PCs.  So the square of the first term is given by
\begin{align*}
\| \widehat \bC_{0,N} -  \tilde \bC_{0,N} \|^2
& = \sum_{i, \ip}^I \sum_{j=1}^{p(i)} \sum_{j=1}^{p(\ip)}
\left( N^{-1} \sum_{n=1}^N \langle X_{i,n} \otimes X_{\ip,n}, \hat v_{j,i} \otimes \hat v_{\jp,\ip} \rangle
-\langle C_{i,\ip},  \hat v_{j,i} \otimes  \hat v_{\jp,\ip} \rangle \right)^2 \\
& \leq \sum_{i, \ip}^I \left\| N^{-1} \sum_{n=1}^N  X_{i,n} \otimes X_{\ip,n}  - C_{i, \ip}\right\|^2
= O_P(N^{-1} I^2).
\end{align*}
The square of the second term is given by
\begin{align*}
 \| \tilde  \bC_{0,N} -  \bC_{0,N}\|^2
 & =  \sum_{i, \ip}^I \sum_{j=1}^{p(i)} \sum_{j=1}^{p(\ip)} \langle C_{i, \ip}, \hat v_{j,i} \otimes \hat v_{\jp,\ip} -  v_{j,i} \otimes v_{\jp,\ip} \rangle^2 \\
& \leq  \sum_{i, \ip}^I \sum_{j=1}^{p(i)} \sum_{j=1}^{p(\ip)} \|C_{i,\ip}\|^2 (\|v_{j,i} - \hat v_{j,i}\| + \|v_{\jp,\ip} - \hat v_{\jp,\ip}\| )^2 \\
& \leq  \sum_{i, \ip}^I \sum_{j=1}^{p(i)} \sum_{j=1}^{p(\ip)} \|C_{i,\ip}\|^2 (\sqrt{2} \alpha_{j,i}^{-1} \| C_{ii} - \widehat C_{i i}\| + \sqrt{2} \alpha_{\jp,\ip}^{-1} \| C_{\ip \ip} - \widehat C_{\ip \ip}\| )^2  \\
& \leq 2 \max \|C_{i,\ip}\|^2 \max \| C_{ii} - \widehat C_{ii} \|^2
\sum_{i, \ip}^I \sum_{j=1}^{p(i)} \sum_{j=1}^{p(\ip)} (\alpha_{j,i}^{-1} + \alpha_{\jp,\ip}^{-1})^2 \\
& = I N^{-1}\Gamma_N O_P(1).
\end{align*}
Therefore both terms are asymptotically bounded by
$I^2 N^{-1} \Gamma_N O_P(1)$, which proves the claim.

\end{enumerate}
\end{proof}
\end{lemma}

{
\begin{proof}[Proof of Theorem \ref{t:multi:HA}]
Analogous results to Lemmas \ref{l:3} and \ref{l:4} are obtained in the same way and are thus omitted for brevity.  We mention that the key difference is that the $\|\Delta_{i,\ip,h}\|$ is no longer of order $O_P(1)$, but $O_P(N^{1/2})$.  The $\widehat C_{ii}$ are still root-$N$ consistent since the series is assumed to be $L^4$-m approximable.  We therefore only show that
\[
\frac{Q_N - p_N^2 H}{p_N \sqrt{2H}} \overset{P}{\to} \infty.
\]
We assume that those lag terms which exhibit correlation are contained in the set $\cH^\star$ we therefore begin with the lower bound
\[
Q_N  \geq  N \sum_{h \in \cH^\star} \bV_{h^\star}^\top ( \bC_0 \otimes \bC_0 )^{-1}\bV_{h^\star}.
\]
Since the smallest eigenvalue of $( \bC_0 \otimes \bC_0 )^{-1}$ is $\gamma_1^{-2}$ we can bound $Q_N$ below using 
\[
 Q_N \geq N \gamma_1^{-2} \sum_{h \in \cH^\star} \| \bV_{h}\|^2.
\]
Isolating the pairs $\cI^\star_h$ which are assumed to be correlated (at a lag of $h$) we can further bound below as
\[
Q_N \geq  \gamma_1^{-2}N \sum_{h \in \cH^\star} \sum_{(i,j) \in \cI_h^\star} \left( N^{-1}\sum_{n=1}^N ( \langle X_{n}, v_{i} \rangle \langle X_{n+h}, v_{j} \rangle \right) ^2
=\gamma_1^{-2}  N R (1+o_P(1))  \sum_{h \in \cH^\star} | \cI_h| 
\]
where the last equality holds since, by Assumption \ref{a:HA}, the summands form a stationary and ergodic sequence.  Combining Lemma \ref{l:bound} with Assumption \ref{a:pn2}, $N \gamma_1^{-2}$ tends to infinity faster than $p_N^2$ and the claim holds.  We also see that the effect of having more indices which exhibit correlation is additive.

\end{proof}
}

\bibliographystyle{apalike}
\renewcommand{\baselinestretch}{0.9}
\small
\bibliography{panel}

\end{document}